\documentclass[final]{l4dc2024}

\usepackage{graphicx}
\usepackage{nicefrac}
\usepackage{mathtools}
\usepackage{dsfont}
\usepackage{hyperref}

\usepackage{tikz}
\usepackage{pgfplots} 
\usepackage{pgfplots}
\pgfplotsset{compat=1.17}
\usepgfplotslibrary{fillbetween}

\definecolor{dgreen}{rgb}{0.0, 0.5, 0.0}
\pgfdeclarelayer{background layer}
\pgfsetlayers{background layer,main}
\def \epst{0.25}%
\def \rhot{4}%
\def \xBound{1.5}%

\DeclareMathOperator*{\argmin}{arg\,min}

\usepackage[nohyperlinks, nolist]{acronym}
\begin{acronym}
\acro{MPC}{model predictive control}
\acro{ISS}{input-to-sate stability}
\end{acronym}

\newtheorem{assumption}{Assumption}
\newtheorem{prop}{Proposition}

\newcommand{\qed}{\hfill \ensuremath{\blacksquare}}

\usepackage{algorithm}
\usepackage{algorithmic}

 \newcommand{\extAppendix}[1]{Appendix~{#1}}
 \newcommand{\refExAppConstraintViolation}{\ref{ex:app_constraintViolation}}
 \newcommand{\refAppProofOfStuff}{\ref{app:proof_of_stuff}}
 \newcommand{\refAppProveLipschitz}{\ref{app:ProveLipschitz}}
 \newcommand{\refSubecAppProofLemma}{\ref{subsec:app_ProofLemma}}
 \newcommand{\refEqAppScmpcDescent}{\eqref{eq:app_SCMPC_descend}}
 \newcommand{\refEqAppInterLemmaHelper}{\eqref{eq:app_InterLemmaHelper2}}

\title[Learning Soft Constrained MPC Value Functions]{Learning Soft Constrained MPC Value Functions: Efficient MPC Design and Implementation providing Stability and Safety Guarantees}

\usepackage{times}

\author{%
 \Name{Nicolas Chatzikiriakos$^{\star, 1}$} \Email{nicolas.chatzikiriakos@ist.uni-stuttgart.de}\\
 \Name{Kim P. Wabersich$^{\star, 2}$} \Email{kimpeter.wabersich@bosch.com}\\
 \Name{Felix Berkel$^{2}$} \Email{felix.berkel@bosch.com}\\
 \Name{Patricia Pauli$^{1}$} \Email{patricia.pauli@ist.uni-stuttgart.de}\\
 \Name{Andrea Iannelli$^{1}$} \Email{andrea.iannelli@ist.uni-stuttgart.de}\\
 \addr $*)$ The first two authors contributed equally to this work \\ $^{1})$ Institute for Systems Theory and Automatic Control, University of Stuttgart,
Germany\\ $^{2})$ Corporate Research Robert Bosch GmbH, 71272 Renningen, Germany  
}

\begin{document}

\maketitle

\begin{abstract}%
    Model Predictive Control (MPC) can be applied to safety-critical control problems, providing closed-loop safety and performance guarantees. Implementation of MPC controllers requires solving an optimization problem at every sampling instant, which is challenging to execute on embedded hardware. To address this challenge, we propose a framework that combines a tightened soft constrained MPC formulation with supervised learning to approximate the MPC value function. This combination enables us to obtain a corresponding optimal control law, which can be implemented efficiently on embedded platforms.
    The framework ensures stability and constraint satisfaction for various nonlinear systems. While the design effort is similar to that of nominal MPC, the proposed formulation provides input-to-state stability (ISS) with respect to the approximation error of the value function. Furthermore, we prove that the value function corresponding to the soft constrained MPC problem is local Lipschitz continuous for Lipschitz continuous systems, even if the optimal control law may be discontinuous. This serves two purposes: First, it allows to relate approximation errors to a sufficiently large constraint tightening to obtain constraint satisfaction guarantees. Second, it paves the way for an efficient supervised learning procedure to obtain a continuous value function approximation. We demonstrate the effectiveness of the method using a nonlinear numerical example.
\end{abstract}


\section{Introduction}
\Ac{MPC} is a technique that can provide a nearly optimal control policies subject to state and input constraints, which has seen significant success in many areas such as autonomous driving \citep{hrovat2012development} or power electronics \citep{karamanakos2020model}. 
In \ac{MPC}, a state measurement is used to predict and optimize the system's future evolution through an optimization problem, often referred to as the \ac{MPC} problem. 
The first element of the optimal control sequence resulting from solving the \ac{MPC} problem is then applied to the system, which is repeated at every time step. 
A key challenge is solving the \ac{MPC} problem with sufficient accuracy in real-time \citep{karamanakos2020model} while considering the demanding implementation in safety-critical systems such as vehicles and airplanes. To overcome this problem, explicit parametric solutions and approximations of the implicit control law, i.e., the mapping from the current state to the first element of the optimal control sequence, have been proposed. 
\par
The basic case of linear quadratic \ac{MPC} problems can, e.g., be reformulated as a multiparametric program, which can be solved offline for all possible system states. The result is an exact solution of the \ac{MPC} problem represented by a piecewise affine mapping \citep{bemporad2002explicit}.
While cheap to evaluate online, such explicit \ac{MPC} methods quickly lead to intractable offline computations as the optimization problem size increases for practically relevant system dimensions and planning horizons~\citep{alessio2009survey}.
More recently, function approximators, such as artificial neural networks, have been used to obtain representations of the implicit MPC control law using supervised learning to additionally cover large-scale and nonlinear systems \citep{lucia2018deep, hertneck2018learning}, which has also been succesfully demonstrated in real-world applications~\citep{nubert2020safe,abu2022deep}.
One approach to certifying desired closed-loop properties is to treat approximation errors as perturbations, allowing the use of existing methods to provide, e.g., constraint satisfaction guarantees. 
Often, such approaches require `sufficient' accuracy in combination with potentially complicated robust MPC designs, which can be challenging to obtain. This becomes particularly relevant in the case of nonlinear \ac{MPC} problems, which typically have discontinuous control mappings~\citep{rawlings2017model} that must be approximated. As a result, systematic errors are introduced whenever a continuous function class, such as a typical artificial neural network, is used for supervised learning. 

\textit{Contributions:} Instead of approximating a potentially \textit{discontinuous} MPC control mapping
\begin{equation}\label{eq:approximate_control_law_directly}
        \tilde u(x) \approx u^*(x) = \argmin ~\textit{MPC objective} \quad \text{s.t.} \quad \textit{MPC constraints}, 
\end{equation}
with symbols ``$u^*(x)$'' and ``$\tilde u(x)$'' representing the optimal and approximate MPC policy at system state ``$x$'', we design and approximate a \textit{continuous} soft constrained MPC value function of the form
\begin{equation}\label{eq:approximate_smpc_value_function}
\tilde V^{\mathrm{s}}(x) \approx V^{*,\mathrm{s}}(x) = \min \textit{MPC objective + slack penalty} ~ \text{s.t.} ~ \textit{tightened MPC soft constraints}.
\end{equation}
Application of a dynamic programming control law~\citep{bertsekas2022abstract} using the approximate value function $\tilde V^{\mathrm{s}}(x)$ resembles an approximate MPC controller, which avoids demanding online optimizations.
To ensure that the MPC value function $V^{*,\mathrm{s}}(x)$ is continuous while achieving desirable closed-loop properties, several modifications are employed to the underlying MPC problem.
Combining a simple constraint tightening with a soft constrained MPC formulation~\citep{kerrigan2000soft} guarantees constraint satisfaction and stability by reflecting potentially unsafe states through large cost values. To this end, the linear soft constrained MPC scheme in~\cite{Zeilinger2010SCMPC} is extended to the \emph{nonlinear} case, and local Lipschitz continuity of the resulting MPC value function is established for Lipschitz continuous system dynamics, positive definite stage costs, and polytopic constraints.
Compared to existing approaches, our modifications do not rely on potentially complicated design procedures using robust MPC techniques, which are difficult to apply to nonlinear control problems.
To address potential difficulties in approximating value functions with extreme curvatures when the soft constraints become active, we further exploit the structural properties of the soft constrained MPC problem. These properties allow a separation into smooth performance and safety value functions, which can be efficiently approximated using supervised learning.

\textit{Related work:} 
In the case of linear MPC, \cite{karg2020efficient} specify requirements for a neural network to enable an accurate approximation of the MPC control law.
Approximating nonlinear MPC problems using neural networks is considered by \cite{hertneck2018learning}, who treat errors through a robust \ac{MPC} design to maintain desirable closed-loop properties.
\cite{drgovna2022differentiable} train a neural network using an \ac{MPC}-inspired cost, including penalties on constraint violations. 
To obtain robustness and safety guarantees, \cite{Drgona2022CBF} extend this approach by a particular type of control barrier function, which can be challenging to obtain in the case of nonlinear and large-scale systems.
The work of \cite{rosolia2017Racing} uses a sampled safe set to approximate the MPC value function based on previous trajectories. The extension proposed by~\cite{wabersich2022SCMPC_Learning} additionally considers the integration of constraint violations using soft constrained MPC and exact penalty functions, providing the conceptual foundation of this work.
Since we derive the proposed control law from an approximated value function, there are apparent connections to approximate dynamic programming \citep{powell2007approximate}, where, e.g., \cite{granzotto2021stop} consider how the approximation error in the optimal value function affects the stability of the closed-loop. 
\paragraph{Outline:} The problem formulation and MPC is introduced in Section~\ref{sec:Preliminaries}. Section~\ref{sec:smpc} presents a reformulation of the standard MPC problem using soft constraints. This reformulation includes constraint penalties into the value function, enabling its use in a dynamic programming control law approximating the MPC control law. Section~\ref{sec:approx_mpc_analysis} discusses this control law in detail. To ensure stability and constraint satisfaction, we establish a relation between the effect of approximation errors of the MPC value function and a simple linear tightening of the softened state constraints. This mechanism is formalized in Section~\ref{subsec:stability} and \ref{subsec:constraint_satisfaction}. In Section~\ref{sec:mpc_value_function_approximation}, we provide an efficient technique to perform the data labeling and supervised learning task to obtain an approximation of the soft constrained value function~\eqref{eq:approximate_smpc_value_function}. Section~\ref{sec:numerical_example} concludes with a numerical example.
\paragraph{Definitions:}
A continuous function $\alpha:[0, a) \to [0, \infty)$ is called a $\mathcal{K}$-function if it is strictly increasing and $\alpha(0) = 0$. It is a $\mathcal{K}_\infty$-function if $a = \infty$ and $\alpha(r) \stackrel{r\to \infty}{\longrightarrow} \infty$. A continuous function $\beta:[0, a) \times [0, \infty) \to [0, \infty)$ is a $\mathcal{KL}$-function if for each fixed $s$, $\beta(r,s)$ is a $\mathcal{K}$-function and for each fixed $r$, $\beta(r,s)$ is decreasing w.r.t. $s$ and $\beta(r,s)\stackrel{s\to \infty}{\longrightarrow} 0$.
\section{Problem Formulation and MPC Background}\label{sec:Preliminaries}
We consider discrete-time, nonlinear systems of the form
\begin{equation}
	x(k+1) = f(x(k),u(k)), \label{eq:system_model}
\end{equation}
with state $x(k)\in \mathbb{R}^{n_x}$, input $u(k) \in \mathbb{R}^{n_u}$ and time index $k \in \mathbb{N}$. We assume that $f: \mathbb{R}^{n_x}\times \mathbb{R}^{n_u} \to \mathbb{R}^{n_x}$ is Lipschitz continuous on bounded sets. System~\eqref{eq:system_model} is assumed to have an equilibrium at the origin, i.e., $0 = f(0,0)$. 
The system is subject to polytopic state-input constraints of the form 
\begin{align}\label{eq:constraints}
  x(k) & \in \mathcal X \triangleq \left\{x \in \mathbb{R}^{n_x} | H_x x \leq \mathds{1} \right\}, \quad u(k) \in \mathcal U \triangleq \left\{u \in \mathbb{R}^{n_u} |H_u u \leq \mathds{1} \right\},
\end{align}
where $H_x \in \mathbb{R}^{m_x \times n_x}$,$H_u \in \mathbb{R}^{m_u \times n_u}$, and $\mathds{1}$ is a vector of ones of appropriate dimension.
While we focus on constraints of the form~\eqref{eq:constraints} for simplicity, our results can be extended to compact constraints of the form $\{(x,u)\in\mathbb R^{n_x} \times \mathbb R^{n_u} | c_x(x) \leq \mathds{1}, c_u(u)\leq \mathds{1}\}$ containing the origin with $c_x, c_u$ Lipschitz continuous.
Given a control objective in the form of minimizing $\Sigma_{k=0}^{\infty} \ell(x(k),u(k))$, with stage cost function $\ell (x,u) = \Vert x \Vert_Q^2 + \Vert u \Vert_R^2$ and weighting matrices $Q \succeq 0$ and $R \succ 0$, MPC provides a principled method for controller design.
At every time step $k$, the system state $x(k)$ is measured and an MPC problem of the following form is solved:
\begin{subequations}
    \begin{alignat}{3}
        V^{*}(x(k)) = &\min_{u_{\cdot \vert k}} && J(x(k), u_{\cdot \vert k}) = V_\mathrm{f}(x_{N \vert k}) + &\sum_{i=0}^{N-1} \ell(x_{i\vert k}, u_{i\vert k})
        \label{eq:MPC_cost}\\
        &\quad \text{s.t. } \quad && x_{0\vert k} = x(k) &\label{eq:MPC_IC} \\
        & && x_{i+1\vert k} = f(x_{i\vert k}, u_{i\vert k})  &\forall i \in [0, N-1]\label{eq:MPC_dyn}\\
        & &&u_{i\vert k} \in \mathcal{U}  &\forall i \in [0, N-1]\label{eq:MPC_input}\\
        & && H_x x_{i\vert k} \le \mathds{1}(1 - \eta)  &\forall i \in [0, N-1]
        \label{eq:MPC_state}\\
        & && x_{N\vert k} \in \bar{\mathcal{X}}_\mathrm{f}, \label{eq:MPC_Terminal} 
    \end{alignat}\label{eq:MPC_problem}%
\end{subequations}
where we optimize over an input and state sequence $\{u_{i\vert k}\}_{i=0}^{N-1}$ and $\{x_{i\vert k}\}_{i=0}^{N}$, with $i$ denoting the prediction time step. The resulting MPC control law is then given by
\begin{equation}\label{eq:MPC_policy}
    \pi_{\mathrm{MPC}}(x(k)) = u^*_{0\vert k}(x(k)),~\text{defined on}~
\mathcal{X}^\eta_N \triangleq \{x(k) \in \mathbb{R}^{n_x} \vert \exists u_{\cdot\vert  k} \text{ s.t. } ~\eqref{eq:MPC_IC} -~\eqref{eq:MPC_Terminal}\},
\end{equation}
with $u^*_{0|k}(x(k))$ denoting a solution to~\eqref{eq:MPC_problem}.
In~\eqref{eq:MPC_cost} we consider a finite sum of $N$ predicted time steps as the objective with an additional terminal cost term $V_{\mathrm f}$, bounding neglected future stage cost terms.
Equation~\eqref{eq:MPC_Terminal} enforces a terminal constraint with $\bar {\mathcal{X}}_\mathrm{f} =\{x\in \mathbb{R}^{n_x} \vert x^\top \bar P x \le \bar 
 h_\mathrm{f} \}$, with $\bar h_\mathrm{f} >0$ and $\bar P \succ 0$. For the case $\eta=0$ in~\eqref{eq:MPC_state}, problem~\eqref{eq:MPC_problem} reduces to a common MPC formulation~\citep{MAYNE2000MPC, rawlings2017model} and provides rigorous constraint satisfaction and asymptotic stability with respect to the origin if the following standard assumption is satisfied.
\begin{assumption}\label{ass:TerminalSetMPC}
    Consider a terminal set $\bar{\mathcal{X}}_\mathrm{f}\subseteq \mathcal X$, $0 \in \bar{\mathcal{X}}_\mathrm{f}$ and a terminal cost function $V_\mathrm{f}: \bar{\mathcal{X}}_\mathrm{f} \to \mathbb R_{\ge 0}$.  There exists a terminal control law $u = K(x)$ such that for all $x \in \bar{\mathcal{X}}_\mathrm{f}$ it holds that $f(x, K(x)) \in  \bar{\mathcal{X}}_\mathrm{f}$, $u = K(x) \in \mathcal U$, and $V_\mathrm{f}(f(x, K(x))) - V_\mathrm{f}(x) \le - \ell(x, K(x))$.
\end{assumption}
In addition to standard MPC, problem~\eqref{eq:MPC_problem} includes a simple constraint tightening factor $\eta \in [0, 1)$ in~\eqref{eq:MPC_state}, which will be used to achieve constraint satisfaction despite approximation errors.
%
\section{Incorporating constraint penalties into value functions using soft constraints}\label{sec:smpc}
State and input constraints are not part of the objective function in classical MPC. Direct application of the dynamic programming control law $\argmin_{u\in\mathcal U} \ell(x(k), u) + V^*(f(x(k),u))$ can therefore lead to constraint violations in many cases, see \extAppendix{\refExAppConstraintViolation} for an illustrative example.
In contrast, the soft-constrained MPC formulation in~\cite{zeilinger2014SCMPC} uses large slack penalties on slack variables in the objective. While the original motivation is to maintain recursive feasibility of the MPC problem even for infeasible states $x(k)$ through softening the constraints, we combine this mechanism with a small but non-zero constraint tightening $\eta>0$ in~\eqref{eq:MPC_problem}.  As a result, constraint violations of the original constraints become visible in the soft constrained MPC value function and can therefore be prevented by design.
In addition, we can exploit the theoretical properties, such as stability and robustness of the original method to derive desired properties in a second step.
Transferring the method provided by~\cite{zeilinger2014SCMPC} to the nominal nonlinear MPC problems~\eqref{eq:MPC_problem} yields the soft constrained MPC problem of the form
\begin{subequations}
	\begin{alignat}{2}
	V^{*, \mathrm{s}}(x(k)) = \min_{u_{\cdot|k}, \xi_{\cdot\vert k}, \alpha } & \quad J(x(k), u_{\cdot \vert k}) + \ell_\xi(\xi_{N\vert k})+ \sum_{i=0}^{N-1} \ell_\xi(\xi_{i\vert k} + \xi_{N\vert k}) \label{eq:SCMPC_Cost} &&\\
	\text{s.t. }    & \quad \eqref{eq:MPC_IC} - \eqref{eq:MPC_input}, \quad 0 \le \alpha \le 1 &&\label{eq:SCMPCconstr_scaling}  \\
                    & \quad x_{N|k} \in \alpha \mathcal X_\mathrm{f}&&\label{eq:SCMPCconstr_Terminal}\\
                    & \quad \alpha \mathcal{X}_\mathrm{f} \subseteq \{x\in\mathbb R^n|H_x x  \le \mathds{1}(1 - \eta)  + \xi_{N|k} \} \label{eq:SCMPCconstr_TerminalCoupling} \\
                    & \quad H_x x_{i|k} \leq  \mathds{1}(1 - \eta)  + \xi_{i|k} + \xi_{N|k} \quad &&\forall i\in [0, N-1]\label{eq:SCMPCconstr_state} \\
                    & \quad 0 \le \xi_{i \vert k} \quad &&\forall i\in [0, N], \label{eq:SCMPCconstr_Slack}
	\end{alignat} \label{eq:SCMPC_Problem_full}%
\end{subequations}
where the optimization variables additionally include the predicted slack variable sequence $\{\xi_{i|k}\}_{i=0}^{N}$, and the terminal set scaling factor $\alpha$. The penalty function $\ell_\xi(\cdot)$ introduces large cost values in the case of (tightened) state constraint violations through the slack variables $\xi_{i \vert k} \in \mathbb{R}^{m_x}$ in~\eqref{eq:SCMPC_Cost} along the horizon $i=0,..,N-1$ . Typically, $\ell_\xi$ is a class $\mathcal K$-function including a norm, which we select as the 1-norm, i.e., $\ell_\xi(\xi) = \rho \Vert \xi \Vert_1$, where the scaling $\rho >0$ is a design parameter. Ideally, $\ell_\xi$ represents an exact penalty function~\citep{kerrigan2000soft} to recover~\eqref{eq:MPC_problem} if $x(k)\in\mathcal X_N^\eta$. The terminal set constraint~\eqref{eq:SCMPCconstr_Terminal} allows for scaling the terminal set $\alpha\mathcal X_{\mathrm f}$, which can yield additional terminal slack penalties through $\xi_{N|k}$ in~\eqref{eq:SCMPCconstr_TerminalCoupling} since $1\cdot \mathcal X_{\mathrm f}$ is not required to be a subset of the state constraints.
\begin{assumption}\label{ass:TerminalSetSCMPC}
    The terminal set $\alpha \mathcal X_{\mathrm f}$ with $\bar{\mathcal X}_f\subset \mathcal X_{\mathrm f}$ satisfies all conditions in Assumption~\ref{ass:TerminalSetMPC} for all $\alpha\in [0,1]$ except for $\mathcal X_{\mathrm f}\subseteq \mathcal X$, i.e., neglecting state constraints.
\end{assumption}
Since the terminal set does not need to satisfy the state constraints, i.e., $\mathcal X_{\mathrm f} \not\subset \mathcal X$, we can enlarge the terminal set, leading to more degrees of freedom compared to the hard constrained case. Different from the linear case~\citep{zeilinger2014SCMPC}, invariance of $\alpha \mathcal X_{\mathrm f}$ is required explicitly for all $\alpha \in [0, 1]$.
Hence, the soft constrained MPC~\eqref{eq:SCMPC_Problem_full} contains an additional assumption to guarantee asymptotic stability of the closed-loop compared with the nominal MPC problem~\eqref{eq:MPC_problem}.
The constraint~\eqref{eq:SCMPCconstr_TerminalCoupling} ensures that the terminal set lies in a softened state constraint set defined by the slack variable $\xi_{N|k}$. This terminal slack is included at every prediction step to ensure that the scaled terminal set is a subset of the softened state constraints. Adding $\xi_{N|k}$ to every prediction step penalty function, the terminal set is selected possibly small through the optimizer, which is crucial in establishing asymptotic stability.
The resulting control law of the soft constrained \ac{MPC} is given by 
\begin{equation}\label{eq:SCMPC_policy}
   \pi_\mathrm{MPC}^\mathrm{s}(x(k)) = u_{0\vert k}^*(x(k)),\text{ defined on } \mathcal{X}^\eta_{N, \mathrm{s}} \triangleq\{x(k) \in\mathbb{R}^{n_x}\vert \exists u_{\cdot \vert k} \text{ s.t. }~\eqref{eq:SCMPCconstr_scaling} -~\eqref{eq:SCMPCconstr_Slack}, \eta > 0\},
\end{equation}
with the following properties resulting from applying \cite{Zeilinger2010SCMPC} to nonlinear systems, see \extAppendix{\refAppProofOfStuff} for a sketch of the proof.
\begin{prop}\label{prop:SCMPCStability}
    Let Assumption\,\ref{ass:TerminalSetSCMPC} hold and consider \eqref{eq:SCMPC_Problem_full} with $\ell (x,u) = \Vert x \Vert_Q^2 + \Vert u \Vert_R^2$. If~\eqref{eq:SCMPC_Problem_full} is initially feasible, then it follows under application of $u(k) = \pi_\mathrm{MPC}^\mathrm{s}(x(k))$ that~\eqref{eq:SCMPC_Problem_full} is recursively feasible and the origin is an asymptotically stable equilibrium point for the closed-loop system~\eqref{eq:system_model}.
\end{prop}
The framework of soft constrained MPC not only reflects states that are close to constraint violations in the value function when selecting $\eta >0$ but also ensures that the corresponding MPC value function~\eqref{eq:SCMPC_Problem_full} is Lipschitz continuous on bounded sets.
\begin{theorem}\label{thm:LipschitzVFcn}
Consider the soft constrained MPC problem~\eqref{eq:SCMPC_Problem_full} and suppose the MPC cost $J(x, u)$ is Lipschitz continuous on bounded sets. Let the same hold for the penalty function $l_\xi(\xi)$ and the system dynamics $f(x,u)$. Then the optimal value function is Lipschitz continuous for all $x \in \mathcal X^\eta_{N,s}$ for which the constraint $\alpha \leq 1$ in~\eqref{eq:SCMPCconstr_scaling} is not active, i.e., it does not affect the optimal solution.
\end{theorem}
The proof can be found in \extAppendix{\refAppProveLipschitz} and is based on eliminating all remaining state constraints into a Lipschitz continuous objective defined on a compact set.
This insight shows that $V^{*,\mathrm{s}}(x)$ in \eqref{eq:SCMPC_Problem_full} is local Lipschitz continuous under common MPC assumptions, and thus, there is a reduced systematic error during supervised learning when using continuous function approximator classes.

\section{Value function-based approximation of the \ac{MPC} control law}\label{sec:approx_mpc_analysis}
This section investigates effects from approximating of the soft constrained value function $V^\mathrm{*,s}(x)$ with $\tilde V^\mathrm{s}(x)$ as outlined in~\eqref{eq:approximate_smpc_value_function} based on the soft constrained \ac{MPC} problem~\eqref{eq:SCMPC_Problem_full}.
The corresponding approximate optimal control policy \citep{bertsekas2022abstract, powell2007approximate} providing the desired approximate MPC controller is given by 
\begin{equation} \label{eq:ApproxControl}
    \tilde u(x(k)) \in \argmin_{u \in \mathcal{U}} \ell(x(k), u) + \tilde V^\mathrm{s}(f(x(k), u)).
\end{equation}
While the controller \eqref{eq:ApproxControl} still requires solving an optimization problem online, the resulting problem is substantially easier to solve efficiently compared with the MPC problem~\eqref{eq:SCMPC_Problem_full} for common parametrizations of the value function, such as artificial neural networks. This increased efficiency can be achieved through gridding and parallel evaluation of the input space. Furthermore, in cases of, e.g., vehicles or airplanes, the input is often low-dimensional compared to a large number of system states, which paths the way for an efficient implementation on embedded hardware.
\par
The upcoming results quantify the effect of approximation errors between $\tilde V^\mathrm{s}(x)$ and the true MPC value function, i.e., $\varepsilon(x) \triangleq \tilde V^\mathrm{s}(x) - V^{*, \mathrm{s}}(x)$, which we assume to be bounded as follows.
\begin{assumption}\label{ass:ErrorBound}
    Consider an approximate soft constrained MPC value function $\tilde V^\mathrm{s}(x): \mathcal{X}^+_{\mathcal{U}, \mathrm{s}} \cup \mathcal{X}_{N, \mathrm{s}}^\eta \rightarrow \mathbb R$ of $V^{*, \mathrm{s}}$, as defined in~\eqref{eq:SCMPC_Problem_full}, 
    with $\mathcal{X}^+_{\mathcal{U}, \mathrm{s}}\triangleq \{x^+ = f(x, u) \vert x  \in \mathcal{X}^\eta_{N, \mathrm{s}},  u \in \mathcal{U}\}$. For all $x \in \mathcal{X}^\eta_{N, \mathrm{s}}$ there exists an $\hat\varepsilon>0$ satisfying
    \begin{equation} \label{eq:errorBound}
        \vert \varepsilon(x) \vert \triangleq \vert \tilde V^\mathrm{s} (x) - V^{*, \mathrm{s}}(x) \vert \le \hat \varepsilon .
    \end{equation}
\end{assumption}
While assuming similar upper bounds is common practice in literature on approximate MPC control laws, see, e.g., \cite{hertneck2018learning}, we present an efficient supervised learning parametrization to obtain possibly small values $\varepsilon(x)$ in Section~\ref{sec:mpc_value_function_approximation}. The corresponding bound $\hat \varepsilon$ can typically be estimated for common function approximator classes such as artificial neural networks or kernel-based learning techniques, see, e.g., \cite{gal2016dropout,maddalena2021deterministic}.

\subsection{Stability}\label{subsec:stability}
To establish \ac{ISS} of the closed-loop system under~\eqref{eq:ApproxControl}, w.r.t.~the approximation error defined in Assumption~\ref{ass:ErrorBound}, we assume invariance of the feasible set of the soft constrained MPC under~\eqref{eq:ApproxControl}.
\begin{assumption}\label{ass:invariance}
    For all $x\in \mathcal{X}^\eta_{N, \mathrm{s}}$ it holds that $f(x, \tilde u(x)) \in \mathcal{X}^\eta_{N, \mathrm{s}}$.
\end{assumption}
While Assumption~\ref{ass:invariance} is a central part of the following investigations, it becomes a technical assumption for Theorem~\ref{th:constrSat} if $\{x^+ = f(x, u) \vert x \in \mathcal{X}, u \in \mathcal{U}\} \subset \mathcal{X}_{N, \mathrm{s}}^\eta$ holds, see paragraph below the proof of Theorem~\ref{th:constrSat}.

%
As a first step toward ISS, we relate the approximate value function $\tilde V^{\mathrm{s}}(x)$ under application of~\eqref{eq:ApproxControl} to the original value function $V^{*,\mathrm{s}}(x)$ under the original control law~\eqref{eq:SCMPC_policy} using~\eqref{eq:errorBound} as follows.
\begin{lemma}\label{lemma:Error_diffU}
If the conditions in Proposition~\ref{prop:SCMPCStability} and Assumptions\, \ref{ass:ErrorBound}, \ref{ass:invariance} hold for some $\eta\in [0, 1) $, then
\begin{equation}
    l(x, \tilde u) + \tilde V^\mathrm{s}(f(x,\tilde u)) \le l(x,\pi_\mathrm{MPC}^\mathrm{s}) + V^{*, \mathrm{s}}(f(x,\pi_\mathrm{MPC}^\mathrm{s})) + \vert \varepsilon(f(x,\pi_\mathrm{MPC}^\mathrm{s}) )\vert \quad \forall x \in \mathcal{X}^\eta_{N, \mathrm{s}}.
\end{equation}
\end{lemma}
A detailed proof can be found in \extAppendix{\refSubecAppProofLemma}. Using Lemma\,\ref{lemma:Error_diffU} we can establish the following system theoretic result, that formalizes how the approximation error influences asymptotic stability of the closed-loop under the approximating control~\eqref{eq:ApproxControl}. More precisely, Theorem\,\ref{th:ISS_ApproxControl} is equivalent to the \ac{ISS} property, which is commonly used in control (see, e.g., \cite{jiang2001input} or \cite{rawlings2017model}) to characterize robust asymptotic stability. 
\begin{theorem}\label{th:ISS_ApproxControl}
If the conditions of Lemma~\ref{lemma:Error_diffU} hold, then it follows for all $x(0) \in \mathcal{X}^\eta_{N, \mathrm{s}}$ that application of the law~\eqref{eq:ApproxControl} implies input-to-state stability, that is,
\begin{equation}
    \vert x(k) \vert \le \beta (\vert x(0)\vert, k) + \sigma\left( \sup_{i \in [0,k-1]} \vert \bar \varepsilon\big(x(i)\big) \vert \right) \quad\forall k \ge 0,
\end{equation}
where $\beta(\cdot)$ is a $\mathcal{KL}$-function, $\sigma(\cdot)$ is a $\mathcal{K}$-function and 
\begin{equation} \label{eq:ISS_error}
   \bar \varepsilon\big(x(k)\big) \triangleq \Big\vert \varepsilon\Big(f\big(x(k), \tilde u\big(x(k)\big)\big)\Big) \Big\vert + \Big\vert  \varepsilon\Big(f\big(x(k), \pi_\mathrm{MPC}^\mathrm{s}\big(x(k)\big)\big)\Big) \Big\vert.
\end{equation}
\end{theorem}
\begin{proof}
By Assumption~\ref{ass:invariance}, the set $\mathcal{X}^\eta_{N, \mathrm{s}}$ is invariant under the control law\,\eqref{eq:ApproxControl}. Further, by Assumption\,\ref{ass:ErrorBound}, we have $\bar\varepsilon \in [0, 2\hat\varepsilon]$, i.e., $ \bar \varepsilon$ lies in a compact set containing the origin. 
Consider now the difference of $V^{*, \mathrm{s}}(f(x, \tilde u))$ and $V^{*, \mathrm{s}}(x)$ for any $x\in \mathcal{X}^\eta_{N, \mathrm{s}}$. First we use~\eqref{eq:errorBound}  and then Lemma\,\ref{lemma:Error_diffU} as well as~\eqref{eq:ISS_error} to obtain\begin{align*}
    V^{*, \mathrm{s}}(f(x, \tilde u)) &- V^{*, \mathrm{s}}(x) \le \tilde V^\mathrm{s}(f(x, \tilde u)) + \vert \varepsilon (f(x, \tilde u))\vert - V^{*, \mathrm{s}}(x )\\ 
    &\le \vert \bar \varepsilon (x) \vert - l(x, \tilde u) + l(x,\pi_\mathrm{MPC}^\mathrm{s}) + V^{*, \mathrm{s}}(f(x,\pi_\mathrm{MPC}^\mathrm{s})) - V^{*, \mathrm{s}}(x ) \\
    & \le - \Vert x \Vert_Q^2 + \vert \bar \varepsilon(x) \vert \leq - \alpha(\Vert x \Vert) + \sigma(\vert \bar \varepsilon\vert ),  
\end{align*}
where $\alpha(\cdot)$ is a $\mathcal{K}_\infty$-function and $\sigma(\cdot)$ is a $\mathcal{K}$-function.  
Note that the last inequality holds by the choice of stage cost function $\ell (x,u) = \Vert x \Vert_Q^2 + \Vert u \Vert_R^2$ in Proposition~\ref{prop:SCMPCStability} and the soft constrained MPC value function decrease under $\pi_{\mathrm{MPC}}$, see~\extAppendix{\refEqAppScmpcDescent}. 
By using similar arguments as in \cite{Zeilinger2010SCMPC}, it can be shown that there exist upper and lower bounding $\mathcal{K}_{\infty}$-functions for $V^{*, \mathrm{s}}(x)$.
Thus, by Definition\,B.46 in \cite{rawlings2017model}, $V^{*, \mathrm{s}}(x)$ is an \ac{ISS}-Lyapunov function for the closed-loop system. Since $\bar \varepsilon$ is bounded by Assumption\,\ref{ass:ErrorBound} applying Lemma\,B.47 in \cite{rawlings2017model} shows the desired result. 
\end{proof}
An important consequence of Theorem\,\ref{th:ISS_ApproxControl} is that the proposed control law \eqref{eq:ApproxControl} renders the origin of system \eqref{eq:system_model} asymptotically stable when there is no approximation error. Further, the closed-loop exhibits the converging-input converging-state property \citep{jiang2001input}, i.e., $\bar \varepsilon \big(x(k)\big) \stackrel{k\to \infty}{\longrightarrow} 0$  implies $x(k) \stackrel{k\to \infty}{\longrightarrow} 0 $.
\subsection{Constraint satisfaction}\label{subsec:constraint_satisfaction}
In order to establish constraint satisfaction, we combine a positive tightening factor $\eta>0$ with ISS properties derived in the previous section. To this end, the following assumption relates the magnitude of the stage cost outside of the tightened state constraints to the approximation error~\eqref{eq:errorBound}.
\begin{assumption} \label{ass:stageCost_Error}
    For all $x \in \mathcal{X}\setminus \mathcal{X}_N^\eta$ it holds that $\ell(x, 0) > 2 \hat \varepsilon$. 
\end{assumption}
Assumption~\ref{ass:stageCost_Error} acts as a normalization between the magnitudes of the MPC value function containing the sum of stage costs and the tolerable approximation error of this value function. In practice, Assumption\,\ref{ass:stageCost_Error} can be satisfied by reducing the value function approximation error as described in Section~\ref{sec:mpc_value_function_approximation}.
%
\begin{theorem}\label{th:constrSat}
   Let the assumptions in Proposition~\ref{prop:SCMPCStability} and Assumptions \ref{ass:ErrorBound}, \ref{ass:invariance}, and~\ref{ass:stageCost_Error} hold for some $\eta\in (0, 1)$  and assume that $\eta \rho \ge V_\mathrm{max}+4\hat \varepsilon$ with slack penalty weight $\rho>0$ and $V_\mathrm{max} \ge V^{*,s}(x)$ $\forall x \in \mathcal{X}^\eta_N$. It follows that application of~\eqref{eq:ApproxControl} implies constraint satisfaction according to~\eqref{eq:constraints} for all trajectories starting in $\mathcal{X}_N^\eta$.
\end{theorem}
\begin{proof}
The core idea is to establish an upper bound on the approximated value function to characterize all states satisfying $\mathcal{X}_{N, \mathrm{s}}^\eta \cap \mathcal{X}$. More specifically, we have for all $x \in \mathcal{X}_{N, \mathrm{s}}^\eta$ that
 \begin{equation} \label{eq:Proof_implicationVmax}
     \tilde V^\mathrm{s}(x) \le V_\mathrm{max} + 3 \hat \varepsilon 
     \implies x \in \mathcal{X}.     
 \end{equation}
To see that this is true, note that $V^{*,s}(x) - \hat \varepsilon \leq \tilde V^\mathrm{s}(x)$ holds due to Assumption~\ref{ass:ErrorBound} and that $\rho$ is chosen such that $V_\mathrm{max}+4\hat \varepsilon\leq \eta \rho $ holds. Combining these facts with the condition in~\eqref{eq:Proof_implicationVmax}, we obtain that the condition in~\eqref{eq:Proof_implicationVmax} implies $V^{*,s}(x)\leq \rho \eta$. From here, \eqref{eq:Proof_implicationVmax} follows by contradiction, i.e., let $V^{*,s}(x)\leq \rho\eta$ and $x\notin \mathcal X$. Due to tightened state constraints \eqref{eq:SCMPCconstr_state}, soft constraint penalty \eqref{eq:SCMPC_Cost} and $J\geq 0$, and $\ell_\xi(\xi)=\rho\Vert \xi \Vert_1$ it follows that for any $x\notin \mathcal X$ we have $V^{*,s}(x)>\eta \rho$, showing the desired contradiction.
Next, we investigate the behavior of the approximation of the value functions along trajectories of the closed-loop under~\eqref{eq:ApproxControl}.
Combining Lemma\,\ref{lemma:Error_diffU} and Assumption\,\ref{ass:ErrorBound} yields
\begin{equation}
		\forall x\in\mathcal{X}_{N, \mathrm{s}}^\eta: \tilde V^\mathrm{s}(f(x,  \tilde{u})) - \tilde V^\mathrm{s}(x)
		\le V^{*, \mathrm{s}}(f(x, \pi_\mathrm{MPC})) + 2 \hat \varepsilon -\ell(x, \tilde u) + \ell(x, \pi_\mathrm{MPC}) - V^{*,\mathrm{s}}(x).
	\end{equation}
The soft constrained MPC value function decrease in \extAppendix{\refEqAppScmpcDescent} then implies $\forall x\in\mathcal{X}_{N, \mathrm{s}}^\eta$:
\begin{align}
    \tilde V^\mathrm{s}(f(x,  \tilde{u})) - \tilde V^\mathrm{s}(x) &\le -\ell(x,\tilde u) + \rho \Vert \xi_N^* \Vert_1 - \rho \Vert \xi_0^* + \xi_N^* \Vert_1 + 2 \hat \varepsilon \\ 
    &\le - \ell(x, \tilde u) - \rho \Vert \xi_0^*\Vert_1 + 2 \hat \varepsilon.\label{eq:ProofHelper}
\end{align}
For any $x \in \mathcal{X}_{N}^\eta$ we have by definition that $\xi^*_0 = 0$, and we can use~\eqref{eq:ProofHelper} to obtain $\tilde V^\mathrm{s}(f(x, \tilde{u})) \le \tilde V^\mathrm{s}(x) + 2\hat\varepsilon \le V_\mathrm{max} + 3\hat\varepsilon$. Together with~\eqref{eq:Proof_implicationVmax} and Assumption~\ref{ass:invariance} this implies that $f(x, \tilde u) \in \mathcal{X} \cap \mathcal{X}_{N,\mathrm{s}}^\eta$. If $f(x,\tilde u) \in \mathcal{X}_{N}^\eta$ the previous argument can be applied recursively.
Otherwise, due to invariance of $\mathcal{X}_{N, \mathrm{s}}^\eta$, the only remaining case during closed-loop is $x \in \mathcal{X} \cap \mathcal{X}_{N, \mathrm{s}}^\eta \setminus \mathcal{X}_N^\eta$. In this case, we use Assumption\,\ref{ass:stageCost_Error} and~\eqref{eq:ProofHelper} to obtain $\tilde V^\mathrm{s}(f(x, \tilde{u})) - \tilde V^\mathrm{s}(x)\le  - \rho \Vert \xi_0^* \Vert_1 < 0$, which, again, together with \eqref{eq:Proof_implicationVmax} implies $f(x,\tilde u) \in  \mathcal{X} \cap \mathcal{X}_{N, \mathrm{s}}^\eta$. Hence, we conclude that $x(0) \in \mathcal{X}_N^\eta \implies x(k) \in \mathcal{X}$ $\forall k$.%
\end{proof}%
Note that in case $\{x^+ = f(x, u) \vert x \in \mathcal{X}, u \in \mathcal{U}\} \subset \mathcal{X}_{N, \mathrm{s}}^\eta$ holds, Assumption~\ref{ass:invariance} can be dropped by successively showing that $x\in\mathcal X$, which equally avoids the necessity of Assumption~\ref{ass:invariance} in Lemma~\ref{lemma:Error_diffU}, i.e. in~\extAppendix{\refEqAppInterLemmaHelper}.
While we use a simple $1$-norm penalty, other penalty functions, e.g., quadratic penalty functions, can also be used by adapting Theorem\,\ref{th:constrSat} accordingly. 
The relation of the scaling $\rho$ with respect to the tightening $\eta$ is crucial to guarantee constraint satisfaction. 
A valid choice allows us to construct a level set of $\tilde V^\mathrm{s}(x)$ which lies inside $\mathcal{X}$.
Figure\,\ref{fig:visuProof} illustrates the cases of a valid and a problematic choice of $\rho$.
\begin{figure}[t]
    \centering
    \begin{tikzpicture}[scale=1.0]
            \begin{axis}[
        axis lines=middle,
        xlabel={$x$},
        ylabel={$V$},
        domain=-2.5:2.5,
        samples=100,
        enlarge x limits=true,
        enlarge y limits=true,
        xtick={-2.5, -2, -1.5, 1.5, 2, 2.5},
        ytick= \empty, 
        xticklabels = \empty,
        yticklabels = \empty,
        ticklabel style={font=\small},
        y=1.4cm/3.3,
        x=1.4cm,
        xmin=-2.3,xmax=2.3,
        ymin=-0.3, ymax = 6,
    ]

    \begin{pgfonlayer}{background layer}
        \fill[color=dgreen!30] (-2,0) rectangle (2,5);
        \draw (1.5, 4.2) node {\textcolor{dgreen}{$\mathcal{X}$}};
        \draw[cyan, ultra thick](-\xBound,-0.04) -- (\xBound,-0.04) node [pos=0.1, below] {\textcolor{cyan}{$\mathcal{X}_{N}^\eta$}};
        \draw[gray, ultra thick](-2.5,0.04) -- (2.5,0.04) node [pos=0.05, below] {\textcolor{gray}{$\mathcal{X}_{N, \mathrm{s}}^\eta$}};
    \end{pgfonlayer}
    
    \addplot[blue, thick, name path=value_function] {
        (x <= 0) * (-1.1 *cos(deg(x)*1.9) + 1.1) +
        (x > 0) * (0.4*x^2) +
        (abs(x) > \xBound) * \rhot *(abs(x)-\xBound)
    };

    \addplot[red, thick, dashdotted, name path=value_functionRho] {
        (x <= 0) * (-1.1 *cos(deg(x)*1.9) + 1.1) +
        (x > 0) * (0.4*x^2) +
        (abs(x) > \xBound) * \rhot/10 *(abs(x)-\xBound)
    };
    
    \addplot[blue, dashed, name path=upper] {
        (x <= 0) * (-1.1 *cos(deg(x)*1.9) + 1.1) +
        (x > 0) * (0.4*x^2) +
        (abs(x) > \xBound) * \rhot*(abs(x)-\xBound) + \epst 
    };
    
    \addplot[blue, dashed, name path=lower] {
        (x <= 0) * (-1.1 *cos(deg(x)*1.9) + 1.1) +
        (x > 0) * (0.4*x^2) +
        (abs(x) > \xBound) * \rhot*(abs(x)-\xBound) - \epst
    };


    \addplot[red, dashed, name path=upperRho] {
        (x <= 0) * (-1.1 *cos(deg(x)*1.9) + 1.1) +
        (x > 0) * (0.4*x^2) +
        (abs(x) > \xBound) * \rhot/10*(abs(x)-\xBound) + \epst 
    };
    
    \addplot[red, dashed, name path=lowerRho] {
        (x <= 0) * (-1.1 *cos(deg(x)*1.9) + 1.1) +
        (x > 0) * (0.4*x^2) +
        (abs(x) > \xBound) * \rhot/10*(abs(x)-\xBound) - \epst
    };

    
    \draw[thick, dotted] (-2,2.25) -- (-0.8,2.25) node [pos=1, right] {\small $V_\mathrm{max}$};
    \draw[thick, dotted] (0,2.25) -- (2 , 2.25);
    \draw[thick, dotted] (-2,3) -- (-1.2,3) node [pos=1, right] {\small $V_\mathrm{max} +3 \hat \varepsilon$};
    \draw[thick, dotted] (0,3) -- (2 ,3);
    \end{axis}
    \end{tikzpicture}
    \caption{Visualization of the proof of Theorem~\ref{th:constrSat}: Optimal value function and corresponding approximation error with appropriate choice of $\rho$ (blue) and $\rho$ selected too small (red). Cases with $\rho$ selected to small can cause constraint violation, see left boundary of $\mathcal X$. The level-set $\{x\in\mathbb R^{n_x} | \tilde V^\mathrm{s}(x) \le V_\mathrm{max} + 3 \hat \varepsilon \}$ bounding closed-loop trajectories has to be contained in $\mathcal{X}$ to guarantee constraint satisfaction.}
    \label{fig:visuProof}
\end{figure}
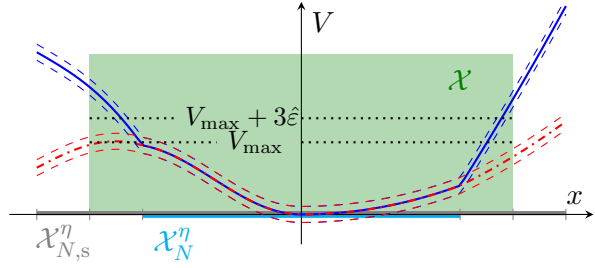
 By using the properties of the soft constrained \ac{MPC} it can then be guaranteed that a level-set w.r.t. $\tilde V^\mathrm{s}(x)$ is invariant, yielding the desired closed-loop constraint satisfaction guarantee. If $\rho$ is not selected appropriately, constraint violations may occur even when $x(0)\in\mathcal{X}_N^\eta$, see \extAppendix{\refExAppConstraintViolation} for an example.
\subsection{Efficient value function approximation by splitting performance and safety}\label{sec:mpc_value_function_approximation}
\begin{figure}
    \includegraphics[width=0.263\linewidth]{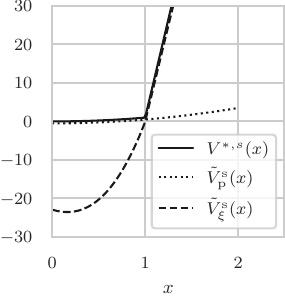}
    \hfill
    \includegraphics[width=0.3\linewidth]{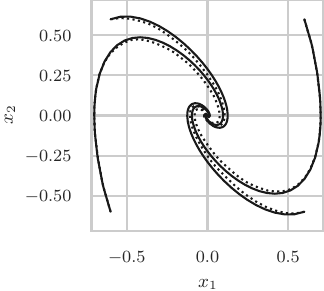}
    \hfill
    \includegraphics[width=0.27\linewidth]{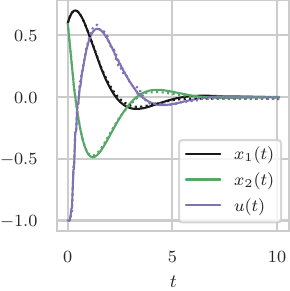}
    \caption{\textit{Left:} Illustrative value function~\eqref{eq:SCMPC_Problem_full} and its approximation using the split defined in~\eqref{eq:value_function_split} to efficiently parametrize sharply increasing curvatures. \textit{Middle/Right:} Simulations according to Section~\ref{sec:numerical_example} under the MPC~\eqref{eq:SCMPC_Problem_full} (solid line) and its approximation~\eqref{eq:ApproxControl} (dashed line), obtained as described in Section~\ref{sec:mpc_value_function_approximation} using 729 data samples. \textit{Middle:} Phase trajectories for different initial conditions. \textit{Right:} Specific state and input trajectory.}\label{fig:illustrations_experiments}
\end{figure}
Theorems~\ref{th:ISS_ApproxControl} and \ref{th:constrSat} rely on a sufficiently accurate approximation of the optimal value function to achieve good convergence and constraint satisfaction guarantees.
The first step toward approximating the optimal MPC value function~\eqref{eq:SCMPC_Problem_full} is to generate a training data set of the form $\mathcal D=\{(x^{(j)}, u^{*,(j)}_{.|0}, \xi^{*,(j)}_{.|0}, x^{*,(j)}_{.|0})\}_{j=1}^{n_{\mathcal D}}$ with $x^{(j)}\in \mathcal X \cap \mathcal X_{N,s}^\eta$ selected by sampling or gridding.
Based on $\mathcal D$, the main approximation challenge we address in the following are rapidly increasing values of $V^{*, \mathrm{s}}(x)$ when $x\notin\mathcal X^\eta_{N}$, i.e., when slack penalties become active. More precisely, a direct approximation of~$V^{*, \mathrm{s}}(x)$ would require function approximators that support steep curvature changes as depicted in Figure~\ref{fig:illustrations_experiments} (left), which typically require large amounts of training samples. 
We avoid this problem by considering the split
\begin{align}\label{eq:value_function_split}
    V^{*, \mathrm{s}}(x) =
    \underbrace{V_{\mathrm f}(x_{N|0}^{*,(j)}) +  \Sigma_{i = 0}^{N-1} l(x_{i|0}^{*,(j)},u_{i|0}^{*,(j)})}_{\approx \tilde V_{\mathrm{p}}^{\mathrm{s}}(x;\theta_{\mathrm{p}})}
    +
    \underbrace{\rho\Vert \xi_{N|0}^{*,(j)} \Vert +\Sigma_{i = 0}^{N-1} \rho\Vert \xi_{i|0}^{*,(j)}+\xi_{N|0}^{*,(j)}\Vert}_{\approx\tilde V_{\mathrm{\xi}}^{\mathrm{s}}(x;\theta_{\mathrm{\xi}})}
\end{align}
into distinctive function approximators $\tilde V_{\mathrm{p}}^{\mathrm{s}}(x;\theta_{\mathrm{p}})$ and $\tilde V_{\mathrm{\xi}}^{\mathrm{s}}(x;\theta_{\mathrm{\xi}})$ with parameters $\theta_{\mathrm{p}}$ and $\theta_{\mathrm{\xi}}$.
The split allows selecting a smooth model class $\tilde V_{\mathrm{p}}^{\mathrm{s}}(x;\theta_{\mathrm{\xi}})$, whose parameters are inferred using $\mathcal D^{\mathrm{p}}=\{(x^{(j)}, V_{\mathrm{p}}^{*, \mathrm{s}}(x^{(j)}))\}_{j=1}^{n_{\mathcal D}}$, which can be generated from $\mathcal D$.
The second term can efficiently be approximated by exploiting non-negativity, i.e., selecting $\tilde V_{\mathrm{\xi}}^{\mathrm{s}}(x;\theta_{\mathrm{\xi}}) = \max(0, \hat V_{\mathrm{\xi}}^{\mathrm{s}}(x;\theta_{\mathrm{\xi}}))$ with some function $\hat V_{\mathrm{\xi}}^{\mathrm{s}}: \mathbb R^{n} \rightarrow \mathbb R$,
which is fitted to $\mathcal D^{\mathrm{\xi}}=\{(x^{(j)}, V_{\mathrm{\xi}}^{*, \mathrm{s}}(x^{(j)}))\}_{j=1}^{n_{\mathcal D}}$.
This reformulation completely circumvents the problem of learning rapidly changing values of $ V^{*, \mathrm{s}}(x)$ as depicted in Figure~\ref{fig:illustrations_experiments} (left). In order to achieve the required approximation accuracy according to Theorem~\ref{th:constrSat}, existing iterative procedures for verification and refinement as proposed, e.g., in \cite[Alg. 1, Steps 4-6]{hertneck2018learning} can be applied.
\section{Numerical Example}\label{sec:numerical_example}
In this section, we illustrate the design steps of the MPC~\eqref{eq:SCMPC_Problem_full} and its approximation~\eqref{eq:ApproxControl} using a nonlinear mass-spring-damper system with state $x(k)=[x_1(k), x_2(k)]^\top \in \mathcal X$, $\mathcal X = [-1,1]^2$ and input $u(k)\in[-1,1]$. The discretized dynamics using Euler-forward with zero-order hold sampling time $T_s = 0.1~s$ are given by $x_1(k+1)=x_1(k) + T_s x_2(k)$, $x_2(k+1) = x_2(k) - T_s(x_1(k)+x_1(k)^3 + 0.1 x_2(k)+u(k))$.
The stage cost function is $\ell(x,u)=x^\top x + u^2$. We construct a terminal set and cost function according to Assumption~\ref{ass:TerminalSetSCMPC} using an LQR controller $K(x)=K_{\mathrm{LQR}}x$ with $\mathcal X_{\mathrm{f}} =\{x\in\mathbb R^2|x^\top P x\leq 1\}$, where $x^\top Px$ is a scaled LQR value function according to~\cite[Alg. 2]{carron2020model}. The planning horizon is $N=50$ and the constraint tightening $\eta$ is selected as $0.3$ with slack penalty $\ell_\xi(\xi)=\rho \xi$ with $\rho=1000$.
A conservative upper bound on $V_{\mathrm{max}}$ can be computed as $N(2 + 1) + \Vert P \Vert=154$, resulting in a minimum approximation accuracy of $\hat \varepsilon \leq 36.5$ to satisfy the conditions in Theorem~\ref{th:constrSat}. Regarding Assumption~\ref{ass:stageCost_Error} we have for all $x\in\mathcal X \setminus \mathcal X_N^\eta$ that $l(x,0)\geq 0.7^2$ resulting in the condition $\hat \varepsilon \leq 0.25$. Following Section~\ref{sec:mpc_value_function_approximation}, we generate a dataset $\mathcal D$ using a uniform grid in $\mathcal X$ with $n_{\mathcal D}=729$ samples.
We chose two artificial neural networks with identical structure to model $\tilde V_{\mathrm{p}}^{\mathrm{s}}(x;\theta_{\mathrm{p}})$ and $\tilde V_{\xi}^{\mathrm{s}}(x;\theta_{\xi})$ with two hidden ReLU layers containing 1024 neurons each. By minimizing the mean-squared error w.r.t. the data sets $\mathcal D^{\mathrm{p}}$ and $\mathcal D^\xi$ using PyTorch~\citep{paszke2017automatic} and its ``adam'' optimizer, we obtained a validation error on $1000$ random samples of $\hat \varepsilon \approx 0.0011$, satisfying the theoretical bound $\hat \varepsilon \leq 0.25$. Evaluation of the approximate control law~\eqref{eq:ApproxControl} is done using a uniform input gridding in the interval $[-1,1]$ with $100$ points, resulting in 1ms  evaluation time on a standard notebook without GPU acceleration.
\par
Figure~\ref{fig:illustrations_experiments} (Middle) illustrates closed-loop simulations under the original MPC controller~\eqref{eq:SCMPC_Problem_full} and its approximation~\eqref{eq:ApproxControl} starting in $\mathcal X_N^\eta$. In addition to ISS and constraint satisfaction according to Theorems~\ref{th:ISS_ApproxControl} and \ref{th:constrSat}, the simulation under the approximate control law displays only small deviations from the original controller without leaving $\mathcal X_N^\eta$. As shown for one specific state and input trajectory in Figure~\ref{fig:illustrations_experiments} (Right), the input trajectories are similar as well, even close to the constraints.

\bibliography{bibliography}

\begin{thebibliography}{27}
\providecommand{\natexlab}[1]{#1}
\providecommand{\url}[1]{\texttt{#1}}
\expandafter\ifx\csname urlstyle\endcsname\relax
  \providecommand{\doi}[1]{doi: #1}\else
  \providecommand{\doi}{doi: \begingroup \urlstyle{rm}\Url}\fi

\bibitem[Abu-Ali et~al.(2022)Abu-Ali, Berkel, Manderla, Reimann, Kennel, and
  Abdelrahem]{abu2022deep}
Mohammad Abu-Ali, Felix Berkel, Maximilian Manderla, Sven Reimann, Ralph
  Kennel, and Mohamed Abdelrahem.
\newblock Deep learning-based long-horizon {MPC}: Robust, high performing, and
  computationally efficient control for {PMSM} drives.
\newblock \emph{IEEE Transactions on Power Electronics}, 37\penalty0
  (10):\penalty0 12486--12501, 2022.

\bibitem[Alessio and Bemporad(2009)]{alessio2009survey}
Alessandro Alessio and Alberto Bemporad.
\newblock A survey on explicit model predictive control.
\newblock \emph{Nonlinear Model Predictive Control: Towards New Challenging
  Applications}, pages 345--369, 2009.

\bibitem[Bemporad et~al.(2002)Bemporad, Morari, Dua, and
  Pistikopoulos]{bemporad2002explicit}
Alberto Bemporad, Manfred Morari, Vivek Dua, and Efstratios~N. Pistikopoulos.
\newblock The explicit linear quadratic regulator for constrained systems.
\newblock \emph{Automatica}, 38\penalty0 (1):\penalty0 3--20, 2002.

\bibitem[Bertsekas(2022)]{bertsekas2022abstract}
Dimitri Bertsekas.
\newblock \emph{Abstract Dynamic Programming: 3rd Edition}.
\newblock Athena scientific optimization and computation series. Athena
  Scientific., 2022.
\newblock ISBN 978-1-886529-47-2.

\bibitem[Boyd and Vandenberghe(2004)]{boyd2004convex}
Stephen~P. Boyd and Lieven Vandenberghe.
\newblock \emph{Convex optimization}.
\newblock Cambridge University Press, 2004.

\bibitem[Carron and Zeilinger(2020)]{carron2020model}
Andrea Carron and Melanie~N Zeilinger.
\newblock Model predictive coverage control.
\newblock \emph{IFAC-PapersOnLine}, 53\penalty0 (2):\penalty0 6107--6112, 2020.

\bibitem[Cortez et~al.(2022)Cortez, Drgona, Tuor, Halappanavar, and
  Vrabie]{Drgona2022CBF}
Wenceslao~S. Cortez, Jan Drgona, Aaron Tuor, Mahantesh Halappanavar, and
  Draguna Vrabie.
\newblock Differentiable predictive control with safety guarantees: A control
  barrier function approach.
\newblock In \emph{2022 IEEE 61st Conference on Decision and Control (CDC)},
  pages 932--938, 2022.
\newblock \doi{10.1109/CDC51059.2022.9993146}.

\bibitem[Drgo{\v{n}}a et~al.(2022)Drgo{\v{n}}a, Ki{\v{s}}, Tuor, Vrabie, and
  Klau{\v{c}}o]{drgovna2022differentiable}
J{\'a}n Drgo{\v{n}}a, Karol Ki{\v{s}}, Aaron Tuor, Draguna Vrabie, and Martin
  Klau{\v{c}}o.
\newblock Differentiable predictive control: Deep learning alternative to
  explicit model predictive control for unknown nonlinear systems.
\newblock \emph{Journal of Process Control}, 116:\penalty0 80--92, 2022.

\bibitem[Gal and Ghahramani(2016)]{gal2016dropout}
Yarin Gal and Zoubin Ghahramani.
\newblock Dropout as a bayesian approximation: Representing model uncertainty
  in deep learning.
\newblock In \emph{international conference on machine learning}, pages
  1050--1059. PMLR, 2016.

\bibitem[Granzotto et~al.(2021)Granzotto, Postoyan, Ne{\v{s}}i{\'c},
  Bu{\c{s}}oniu, and Daafouz]{granzotto2021stop}
Mathieu Granzotto, Romain Postoyan, Dragan Ne{\v{s}}i{\'c}, Lucian
  Bu{\c{s}}oniu, and Jamal Daafouz.
\newblock When to stop value iteration: stability and near-optimality versus
  computation.
\newblock In \emph{Learning for Dynamics and Control}, pages 412--424. PMLR,
  2021.

\bibitem[Hertneck et~al.(2018)Hertneck, K{\"o}hler, Trimpe, and
  Allg{\"o}wer]{hertneck2018learning}
Michael Hertneck, Johannes K{\"o}hler, Sebastian Trimpe, and Frank
  Allg{\"o}wer.
\newblock Learning an approximate model predictive controller with guarantees.
\newblock \emph{IEEE Control Systems Letters}, 2\penalty0 (3):\penalty0
  543--548, 2018.

\bibitem[Hrovat et~al.(2012)Hrovat, Di~Cairano, Tseng, and
  Kolmanovsky]{hrovat2012development}
Davor Hrovat, Stefano Di~Cairano, H.~Eric Tseng, and Ilya~V. Kolmanovsky.
\newblock The development of model predictive control in automotive industry: A
  survey.
\newblock In \emph{2012 IEEE International Conference on Control Applications},
  pages 295--302, 2012.

\bibitem[Jiang and Wang(2001)]{jiang2001input}
Zhong-Ping Jiang and Yuan Wang.
\newblock Input-to-state stability for discrete-time nonlinear systems.
\newblock \emph{Automatica}, 37\penalty0 (6):\penalty0 857--869, 2001.

\bibitem[Karamanakos et~al.(2020)Karamanakos, Liegmann, Geyer, and
  Kennel]{karamanakos2020model}
Petros Karamanakos, Eyke Liegmann, Tobias Geyer, and Ralph Kennel.
\newblock Model predictive control of power electronic systems: Methods,
  results, and challenges.
\newblock \emph{IEEE Open Journal of Industry Applications}, 1:\penalty0
  95--114, 2020.

\bibitem[Karg and Lucia(2020)]{karg2020efficient}
Benjamin Karg and Sergio Lucia.
\newblock Efficient representation and approximation of model predictive
  control laws via deep learning.
\newblock \emph{IEEE Transactions on Cybernetics}, 50\penalty0 (9):\penalty0
  3866--3878, 2020.

\bibitem[Kerrigan and Maciejowski(2000)]{kerrigan2000soft}
Eric~C. Kerrigan and Jan~M. Maciejowski.
\newblock Soft constraints and exact penalty functions in model predictive
  control.
\newblock In \emph{Proceedings of United Kingdom Automatic Control Council
  (UKACC) International Conference on Control}. IET, 2000.

\bibitem[Lucia and Karg(2018)]{lucia2018deep}
Sergio Lucia and Benjamin Karg.
\newblock A deep learning-based approach to robust nonlinear model predictive
  control.
\newblock \emph{IFAC-PapersOnLine}, 51\penalty0 (20):\penalty0 511--516, 2018.

\bibitem[Maddalena et~al.(2021)Maddalena, Scharnhorst, and
  Jones]{maddalena2021deterministic}
Emilio~T. Maddalena, Paul Scharnhorst, and Colin~N. Jones.
\newblock Deterministic error bounds for kernel-based learning techniques under
  bounded noise.
\newblock \emph{Automatica}, 134:\penalty0 109896, 2021.

\bibitem[Mayne et~al.(2000)Mayne, Rawlings, Rao, and Scokaert]{MAYNE2000MPC}
David~Q. Mayne, James~B. Rawlings, Christopher~V. Rao, and Pierre~O.M.
  Scokaert.
\newblock Constrained model predictive control: Stability and optimality.
\newblock \emph{Automatica}, 36\penalty0 (6):\penalty0 789--814, 2000.
\newblock \doi{https://doi.org/10.1016/S0005-1098(99)00214-9}.

\bibitem[Nubert et~al.(2020)Nubert, K{\"o}hler, Berenz, Allg{\"o}wer, and
  Trimpe]{nubert2020safe}
Julian Nubert, Johannes K{\"o}hler, Vincent Berenz, Frank Allg{\"o}wer, and
  Sebastian Trimpe.
\newblock Safe and fast tracking on a robot manipulator: Robust mpc and neural
  network control.
\newblock \emph{IEEE Robotics and Automation Letters}, 5\penalty0 (2):\penalty0
  3050--3057, 2020.

\bibitem[Paszke et~al.(2017)Paszke, Gross, Chintala, Chanan, Yang, DeVito, Lin,
  Desmaison, Antiga, and Lerer]{paszke2017automatic}
Adam Paszke, Sam Gross, Soumith Chintala, Gregory Chanan, Edward Yang, Zachary
  DeVito, Zeming Lin, Alban Desmaison, Luca Antiga, and Adam Lerer.
\newblock Automatic differentiation in {PyTorch}.
\newblock 2017.

\bibitem[Powell(2007)]{powell2007approximate}
Warren~B. Powell.
\newblock \emph{Approximate Dynamic Programming: Solving the curses of
  dimensionality}, volume 703.
\newblock John Wiley \& Sons, 2007.

\bibitem[Rawlings et~al.(2017)Rawlings, Mayne, and Diehl]{rawlings2017model}
James~B. Rawlings, David~Q. Mayne, and Moritz Diehl.
\newblock \emph{Model predictive control: theory, computation, and design},
  volume~2.
\newblock Nob Hill Publishing Madison, WI, 2017.

\bibitem[Rosolia et~al.(2017)Rosolia, Carvalho, and
  Borrelli]{rosolia2017Racing}
Ugo Rosolia, Ashwin Carvalho, and Francesco Borrelli.
\newblock Autonomous racing using learning model predictive control.
\newblock In \emph{2017 American Control Conference (ACC)}, pages 5115--5120,
  2017.
\newblock \doi{10.23919/ACC.2017.7963748}.

\bibitem[Wabersich et~al.(2022)Wabersich, Krishnadas, and
  Zeilinger]{wabersich2022SCMPC_Learning}
Kim~P. Wabersich, Raamadaas Krishnadas, and Melanie~N. Zeilinger.
\newblock A soft constrained {MPC} formulation enabling learning from
  trajectories with constraint violations.
\newblock \emph{IEEE Control Systems Letters}, 6:\penalty0 980--985, 2022.
\newblock \doi{10.1109/LCSYS.2021.3087968}.

\bibitem[Zeilinger et~al.(2010)Zeilinger, Jones, and
  Morari]{Zeilinger2010SCMPC}
Melanie~N. Zeilinger, Colin~N. Jones, and Manfred Morari.
\newblock Robust stability properties of soft constrained {MPC}.
\newblock In \emph{49th IEEE Conference on Decision and Control (CDC)}, pages
  5276--5282, 2010.
\newblock \doi{10.1109/CDC.2010.5717488}.

\bibitem[Zeilinger et~al.(2014)Zeilinger, Morari, and
  Jones]{zeilinger2014SCMPC}
Melanie~N. Zeilinger, Manfred Morari, and Colin~N. Jones.
\newblock Soft constrained model predictive control with robust stability
  guarantees.
\newblock \emph{IEEE Transactions on Automatic Control}, 59\penalty0
  (5):\penalty0 1190--1202, 2014.
\newblock \doi{10.1109/TAC.2014.2304371}.

\end{thebibliography}

\appendix

\section{Proofs}\label{app:proof_of_stuff}
\subsection{Proof of Proposition~\ref{prop:SCMPCStability}}
Due to Assumption\,\ref{ass:TerminalSetSCMPC} the scaled terminal set is invariant under the terminal controller for all $\alpha\in[0,1]$. Thus, one can use the same candidate function as in the proof in \cite{Zeilinger2010SCMPC} with the non-linear terminal control law to show that the value function satisfies 
\begin{equation} \label{eq:app_SCMPC_descend}
            V^{*, \mathrm{s}}(f(x,\pi_\mathrm{MPC}^\mathrm{s}(x))) - V^{*, \mathrm{s}}(x) \le - \ell(x,\pi_\mathrm{MPC}^\mathrm{s}(x)) + \underbrace{\ell_\xi(\xi^*_{N \vert k}) - \ell_\xi(\xi_{0 \vert k}^*+ \xi_{N \vert k}^*)}_{\le 0}.
\end{equation}
From here, asymptotic stability of the origin of the closed-loop can be concluded using standard MPC theory~\citep{rawlings2017model}. \qed
 
\subsection{Proof of Theorem\,\ref{thm:LipschitzVFcn}} \label{app:ProveLipschitz}
\begin{lemma}\label{lem:app_alternativeCoupling}
    The constraint~\eqref{eq:SCMPCconstr_TerminalCoupling} can be stated explicitly as
        \begin{equation}\label{eq:app_SCMPCalternativeCoupling}
            \sqrt{\alpha h_\mathrm{f}} \sqrt{c_j} \le [1 - \eta + \xi_{N\vert k}]_j  \quad \forall j = 1, \dots, n_x,
        \end{equation}
        where $c_j = [H_x]_jP^{-1}[H_x]_j^\top$ and $[\cdot]_j$ selects the $j$-th row of its argument.
\end{lemma}
\begin{proof}
    Similarly as done in \cite{wabersich2022SCMPC_Learning} we use the support function 
    \begin{equation}\label{eq:app_supportFCN}
        h([H_x]_j, \alpha) \triangleq \max_{x \in \mathbb{R}^{n_x}} [H_x]_jx \quad \text{s.t. }  x^\top P x \le \alpha h_\mathrm{f} 
    \end{equation}
    to translate the subset condition~\eqref{eq:SCMPCconstr_TerminalCoupling} into 
    \begin{equation*}
        h([H_x]_j, \alpha) \le [1 -\eta + \xi_{N\vert k}]_j \quad \forall j=1,\dots,n_x.
    \end{equation*}
    We show Lemma\,\ref{lem:app_alternativeCoupling} by deriving the analytical solution the optimization problem in~\eqref{eq:app_supportFCN}.
	The Lagrange function to the optimization problem~\eqref{eq:app_supportFCN} is given by 
	\begin{equation}
	L(x, \nu) = - [H_x]_j x + \nu (x^\top P x - \alpha h_\mathrm{f}).
	\end{equation}
	The dual function is given by $d(\nu) = \min_x L(x, \nu)$, with $\nu \ge 0$.
    Plugging in the minimizer $x^* = \frac{1}{2 \nu} P^{-1} [H_x]_j^\top$ yields
    \begin{equation}
	d(\nu) = - \frac{1}{4 \nu} [H_x]_j P^{-1} [H_x]_j^\top -\nu \alpha h_\mathrm{f}.
	\end{equation}
    By minimizing the dual function with respect to $\nu$ we obtain 
    \begin{equation}
        d(\nu^*)= -  \sqrt{\alpha h_\mathrm{f}} \sqrt{c_j}  \label{eq:app_DualSol_helper1},
    \end{equation}
    where $c_j = [H_x]_j P^{-1} [H_x]_j^\top$.
    Since $h_\mathrm{f}>0$ it follows that the interior of the terminal set $\mathcal{X}_\mathrm{f}$ is non-empty. Thus, because the primal problem~\eqref{eq:app_supportFCN} is convex, we can use Slater's condition to show that strong duality holds and, therefore, the optimal value of the dual problem is equal to the optimal value of the primal problem~\citep{boyd2004convex}.
 \end{proof}
From here, the proof of Theorem\,\ref{thm:LipschitzVFcn} follows from eliminating all relevant state constraints into~\eqref{eq:SCMPC_Cost} as follows.
By Lemma\,\ref{lem:app_alternativeCoupling} we can substitute the constraint~\eqref{eq:SCMPCconstr_TerminalCoupling} by~\eqref{eq:app_SCMPCalternativeCoupling}.
Now we can rewrite~\eqref{eq:SCMPC_Problem_full} by plugging~\eqref{eq:SCMPCconstr_scaling} (except for $\alpha \leq 1$),~\eqref{eq:SCMPCconstr_Terminal},~\eqref{eq:SCMPCconstr_state}  and~\eqref{eq:app_SCMPCalternativeCoupling} into~\eqref{eq:SCMPC_Cost}. The resulting optimization problem is equivalent for cases, where the optimal solution satisfies $\alpha \leq 1$ in~\eqref{eq:SCMPCconstr_scaling}, but it is no longer state constrained. Thus we can apply Theorem C.29 from \cite{rawlings2017model} to show Lipschitz continuity of the value function.

\subsection{Proof of Lemma\,\ref{lemma:Error_diffU}} \label{subsec:app_ProofLemma}
\begin{proof}
	To prove Lemma\,\ref{lemma:Error_diffU} we start from~\eqref{eq:ApproxControl} and relate the approximation of the optimal value function $\tilde V^\mathrm{s}(x)$ to the true optimal value function $V^{*, \mathrm{s}}(x)$ by using the approximation error~\eqref{eq:errorBound} and choosing $u^*(x)=\pi_{\mathrm{MPC}}^{\mathrm{s}}(x)$~\eqref{eq:SCMPC_policy} as a feasible but suboptimal solution to the dynamic programming control law.
	The optimal value of the optimization problem underlying the approximate control law~\eqref{eq:ApproxControl} can be reformulated and upper bounded as  
	\begin{align}
	l(x, \tilde u) + \tilde V^\mathrm{s} (f(x, \tilde u)) &= \min_{u \in \mathcal{U}} l(x, u) + \tilde V^\mathrm{s} (f(x, u)) \label{eq:app_InterLemmaHelper1}\\
	&= \min_{\substack{u \in \mathcal{U} \\ f(x, u) \in \mathcal{X}^\eta_{N, \mathrm{s}}}}  l(x, u) + \tilde V^\mathrm{s} (f(x, u)) \label{eq:app_InterLemmaHelper2}\\
	&\le \min_{\substack{u \in \mathcal{U} \\ f(x, u) \in \mathcal{X}^\eta_{N, \mathrm{s}}}}  l(x, u) +  V^{*,\mathrm{s}} (f(x, u)) + \vert \varepsilon(f(x, u)) \vert  \label{eq:app_InterLemmaHelper3}.
	\end{align}
	Hereby~\eqref{eq:app_InterLemmaHelper1} and~\eqref{eq:app_InterLemmaHelper2} are equal because the minimizer lies in $\mathcal{X}_{N, \mathrm{s}}^\eta$ by Assumption\,\ref{ass:invariance}. The upper bound in~\eqref{eq:app_InterLemmaHelper3} follows using the error definition in~\eqref{eq:errorBound}. Since the true value function $V^{*, \mathrm{s}}$ is only defined on $\mathcal{X}^\eta_{N, \mathrm{s}}$, it is necessary to introduce the second constraint $f(x, u) \in \mathcal{X}^\eta_{N, \mathrm{s}}$ for the objective function to be defined for every feasible optimization variable. 
	Since recursive feasibility of~\eqref{eq:SCMPC_Problem_full} is ensured by Proposition~\ref{prop:SCMPCStability}, it holds that
    \begin{equation}
		x \in \mathcal{X}^\eta_{N, \mathrm{s}} \implies f(x, u^*(x)) \in \mathcal{X}^\eta_{N, \mathrm{s}},
	\end{equation}
 and one can upper bound~\eqref{eq:app_InterLemmaHelper3} by plugging in the suboptimal candidate solution $u^*(x)$ into~\eqref{eq:app_InterLemmaHelper3} to obtain 
	\begin{equation}
	\begin{aligned}
	l(x, \tilde u) + \tilde V^\mathrm{s}(&f(x, \tilde u)) \\ &\le l(x, u^*) + V^{*, \mathrm{s}}(f(x, u^*)) + \vert \varepsilon(f(x, u^*)) \vert,
	\end{aligned}
	\end{equation}
	which concludes the proof.
\end{proof}

\section{Example of constraint violations}\label{ex:app_constraintViolation}
The following example illustrates, why the proposed control scheme \eqref{eq:ApproxControl} can lead to constraint violations when no constraint tightening is deployed.
	Consider a linear time-invariant system with state-space dimension $n = 2$ and input dimension $m = 1$ given by 
	\begin{equation}
		\dot x = Ax + Bu = \begin{bmatrix}
		0 & 0 \\ 0 & -1 
		\end{bmatrix} x + \begin{bmatrix}
		1 \\ 2
		\end{bmatrix} u.
	\end{equation}
	The system is discretized using the Euler method with step size $h = 0.01$ to obtain the discrete-time system
	\begin{equation}
		x(k+1) = (I + h A) x(k) + h B u(k). 
	\end{equation}
	The absolute value of each of the states is constrained to be smaller than $1$, i.\,e.,
	\begin{equation}\label{eq:app_ex1_stateconstr}
		x \in \mathcal{X} = \Big\{x \in \mathbb{R}^2: x \in [-1, 1]\times [-1, 1] \Big\}.
	\end{equation}
	The input is constrained by $u \in [-10, 10]$. 
	The stage cost is chosen quadratic, i.\,e., $l(x, u) = \Vert x \Vert_Q^2 + \Vert u \Vert_R^2$ with $Q = I$ and $R = 1$. To guarantee recursive feasibility and stability of the \ac{MPC} a terminal controller and a terminal set as well as a terminal cost are needed. The terminal controller is derived as a linear quadratic regulator with the same stage cost as the \ac{MPC}.
    Solving the LQR problem also yields the matrix $P_\mathrm{LQR}$. The terminal cost can then be constructed as $V_{\mathrm{f}}(x) =  x^\top P_{\mathrm{LQR}} x$ and yields an asymptotically stable closed-loop. Further, the terminal set is constructed using the LMI solver SeDuMi, to obtain the terminal set $\mathcal{X}_{\mathrm{f}} = \{x: x^\top P x \le 1\}$ with $P = I$. To obtain the soft constrained \ac{MPC} setup \eqref{eq:SCMPC_Problem_full}, the enlarged terminal set with $P = \frac{1}{4} I $ is used. That is, the set of feasible state configurations of the hard constrained problem is fully contained in the enlarged terminal set. Further, the horizon of the \ac{MPC} is chosen as $N = 80$. Choosing the large horizon makes sure that every state configuration satisfying the state constraints \eqref{eq:app_ex1_stateconstr} is feasible for the hard constrained problem, i.\,e. $\mathcal{X}_N = \mathcal{X}$.
	A linear penalty function of the form 
	\begin{equation}
		l_\xi(\xi) = \rho \Vert \xi \Vert_1
	\end{equation}
with penalty weight $\rho = 5$ is used.
Figure\,\ref{fig:app_constrVio} shows the closed-loop trajectories of the \ac{MPC}, soft constrained \ac{MPC} as well as the control law 
\begin{equation}
	u(x)^* = \argmin_u l(x, u) + V^{*, \mathrm{s}}(f(x, u)) \label{eq:app_ex1_AlterControlLaw},
\end{equation} 
starting at the initial condition $x(0) = \begin{bmatrix}
-1 & 1
\end{bmatrix}^\top$. Observe that in \eqref{eq:app_ex1_AlterControlLaw} the value function of the soft constrained \ac{MPC} problem is used instead of an approximation, in order to show that even in the nominal case constraint violations can arise.
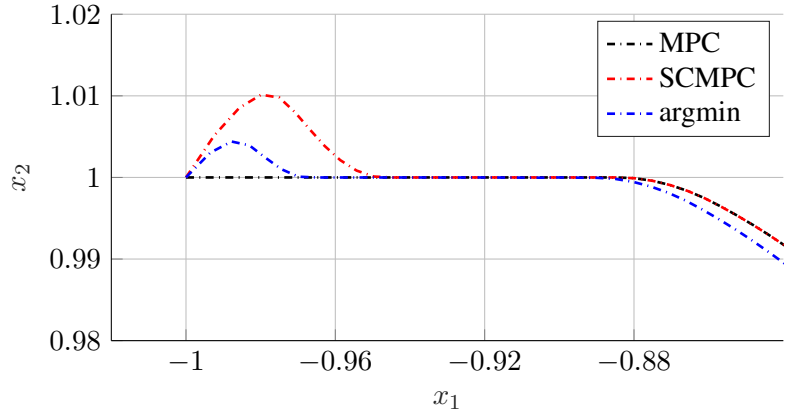
\begin{figure}[ht]
	\centering
        \begin{tikzpicture}[scale=1.0]
%
%

%

\begin{axis}[%
width=3.5in,
height=1.7in,
at={(2.6in,1.103in)},
scale only axis,
xmin=-1.02,
xmax=-0.84,
xlabel style={font=\color{white!15!black}},
xlabel={$x_1$},
ymin=0.98,
ymax=1.02,
ylabel style={font=\color{white!15!black}},
ylabel={$x_2$},
axis background/.style={fill=white},
axis x line*=bottom,
axis y line*=left,
xmajorgrids,
ymajorgrids,
xtick = {-1, -0.96, -0.92, -0.88},
legend style={legend cell align=left, align=left, draw=white!15!black}
]
\addplot [color=black, dashdotted, line width=1pt]
  table[row sep=crcr]{%
-1	1\\
-0.955000014976204	1.00000001228682\\
-0.950003960252132	0.999992121612095\\
-0.945000038165639	1.00000004456896\\
-0.885000083620612	1.00000000542287\\
-0.880042208048864	0.999915756512136\\
-0.875197972653161	0.999605069738421\\
-0.870458906311463	0.999087151724433\\
-0.865817372498886	0.998379347832343\\
-0.861266487618941	0.997497324113908\\
-0.856800039338927	0.996455247432797\\
-0.852378535242146	0.995333703152033\\
-0.84799626873179	0.994144899141223\\
-0.843652576347091	0.992890834919209\\
-0.839346809219254	0.991573460825692\\
};
\label{fig:Ex_MPC}
\addlegendentry{MPC}

\addplot [color=red, dashdotted, line width=1pt]
  table[row sep=crcr]{%
-1	1\\
-0.992429470965891	1.00514105806822\\
-0.985709583995503	1.00852942142831\\
-0.979861018699345	1.01014125780634\\
-0.974989596417555	1.00978268979186\\
-0.970831129094365	1.00800179754032\\
-0.96684829018235	1.00588745738895\\
-0.962837736951202	1.00384968927736\\
-0.958682362364687	1.00212194155761\\
-0.95431168744385	1.00084207198371\\
-0.949673657040223	1.00010971207113\\
-0.944728244828824	0.999999439373215\\
-0.884728003667275	0.999999930950447\\
-0.8797765325392	0.999902873897092\\
-0.874938222223655	0.99958046578921\\
-0.870204673837056	0.999051757904517\\
-0.86556823610681	0.998334115785964\\
-0.861022091851636	0.997443063138451\\
-0.856560056783393	0.996392702643553\\
-0.852140679270612	0.99526753064268\\
-0.847760503187024	0.99407520750343\\
-0.843418863267199	0.992817735268045\\
-0.839115114422013	0.991497055605736\\
};
\label{fig:Ex_SCMPC}
\addlegendentry{SCMPC}

\addplot [color=blue, dashdotted, line width=1pt]
  table[row sep=crcr]{%
-1	1\\
-0.99348	1.00304\\
-0.987784	1.0044016\\
-0.983	1.003925584\\
-0.978732	1.00242232816\\
-0.974428	1.0010061048784\\
-0.96988	1.00009204382962\\
-0.964928	0.99999512339132\\
-0.934928	0.999995408769731\\
-0.929932	0.999987454682034\\
-0.924928	0.999995580135214\\
-0.919928	0.999995624333861\\
-0.90494	0.999971993495521\\
-0.899928	0.999996273560566\\
-0.89494	0.999972310824961\\
-0.889948	0.999956587716711\\
-0.885016	0.999821021839544\\
-0.8802	0.999454811621148\\
-0.875488	0.998884263504937\\
-0.870872	0.998127420869887\\
-0.866344	0.997202146661188\\
-0.8619	0.996118125194577\\
-0.857504	0.994948943942631\\
-0.853144	0.993719454503204\\
-0.848824	0.992422259958172\\
-0.84454	0.991066037358591\\
-0.840292	0.989651376985005\\
-0.83608	0.988178863215155\\
};
\label{fig:ExApproxControl}
\addlegendentry{argmin}

\end{axis}
        \end{tikzpicture}
	\caption{Trajectories of the closed-loop with \acs*{MPC} (\ref{fig:Ex_MPC}), soft constrained \acs*{MPC} (\ref{fig:Ex_SCMPC}) and approximating control law \eqref{eq:app_ex1_AlterControlLaw} (\ref{fig:ExApproxControl}).}
	\label{fig:app_constrVio}
\end{figure} 
As to be expected, the closed-loop trajectory of the \ac{MPC} remains within the state constraints and approaches the origin. 
While the soft constrained \ac{MPC} trajectory initially violates the state constraints initially its trajectory approaches the trajectory of the \ac{MPC} after a few time steps and shares the same trajectory from then on. 
The approximating control law \eqref{eq:app_ex1_AlterControlLaw} like the soft constrained \ac{MPC} violates the constraints but has a slightly different trajectory overall.  
This difference can be explained by the differing optimization problems solved to obtain the control laws.  
\end{document}